\documentclass[11pt,aps,pra,notitlepage,tightenlines,nofootinbib,superscriptaddress]{revtex4-2}

\usepackage{newpxtext,newpxmath}

\let\coloneqq\relax

\usepackage[latin1]{inputenc}
\usepackage{amsthm}
\usepackage{amssymb}
\usepackage{amsmath}
\usepackage{bbold}
\usepackage{bbm}
\usepackage[pdftex, backref=page]{hyperref}
\usepackage{braket}
\usepackage{dsfont}
\usepackage{mathdots}
\usepackage{mathtools}
\usepackage{enumerate}
\usepackage[shortlabels]{enumitem}
\usepackage{csquotes}
\usepackage{stmaryrd}
\usepackage[cal=boondox]{mathalfa}
\usepackage{graphicx}
\usepackage{stackengine}
\usepackage{scalerel}
\usepackage{tensor}       
\usepackage{array}
\usepackage{makecell}
\newcolumntype{x}[1]{>{\centering\arraybackslash}p{#1}}
\usepackage{tikz}
\usepackage{pgfplots}
\usetikzlibrary{shapes.geometric, shapes.misc, positioning, arrows, arrows.meta, decorations.pathreplacing, decorations.pathmorphing, patterns, angles, quotes, calc}
\usepackage{booktabs}
\usepackage{xfrac}
\usepackage{siunitx}
\usepackage{centernot}
\usepackage{comment}
\usepackage{chngcntr}
\usepackage{caption}
\usepackage{subcaption}

\newtheorem{thm}{Theorem}
\newtheorem*{thm*}{Theorem}
\newtheorem{prop}[thm]{Proposition}
\newtheorem*{prop*}{Proposition}
\newtheorem{lemma}[thm]{Lemma}
\newtheorem*{lemma*}{Lemma}

\newtheorem*{cor*}{Corollary}
\newtheorem{cj}[thm]{Conjecture}
\newtheorem*{cj*}{Conjecture}
\newtheorem{Def}[thm]{Definition}
\newtheorem*{Def*}{Definition}

\newtheorem*{question*}{Question}

\newtheorem*{problem*}{Problem}

\makeatletter
\def\thmhead@plain#1#2#3{%
  \thmname{#1}\thmnumber{\@ifnotempty{#1}{ }\@upn{#2}}%
  \thmnote{ {\the\thm@notefont#3}}}
\let\thmhead\thmhead@plain
\makeatother

\theoremstyle{definition}
\newtheorem{rem}[thm]{Remark}
\newtheorem*{note}{Note}

\newcommand{\bb}{\begin{equation}\begin{aligned}\hspace{0pt}}
\newcommand{\bbb}{\begin{equation*}\begin{aligned}}
\newcommand{\ee}{\end{aligned}\end{equation}}
\newcommand{\eee}{\end{aligned}\end{equation*}}
\newcommand*{\coloneqq}{\mathrel{\vcenter{\baselineskip0.5ex \lineskiplimit0pt \hbox{\scriptsize.}\hbox{\scriptsize.}}} =}

\newcommand{\eqt}[1]{\stackrel{\mathclap{\scriptsize \mbox{#1}}}{=}}
\newcommand{\leqt}[1]{\stackrel{\mathclap{\scriptsize \mbox{#1}}}{\leq}}

\newcommand{\ketbra}[1]{\ket{#1}\!\!\bra{#1}}

\newcommand{\ketbraa}[2]{\ket{#1}\!\!\bra{#2}}

\newcommand{\sumno}{\sum\nolimits}

\newcommand{\e}{\varepsilon}
\renewcommand{\epsilon}{\varepsilon}

\newcommand{\id}{\mathds{1}}
\newcommand{\R}{\mathds{R}}
\newcommand{\N}{\mathds{N}}

\newcommand{\C}{\mathds{C}}

\newcommand{\locc}{\mathrm{LOCC}}
\newcommand{\sep}{\mathrm{SEP}}
\newcommand{\ppt}{\mathrm{PPT}}

\DeclareMathOperator{\Tr}{Tr}

\DeclareMathAlphabet{\pazocal}{OMS}{zplm}{m}{n}

\newcommand{\TT}{\pazocal{T}}

\newcommand{\lsmatrix}{\left(\begin{smallmatrix}}
\newcommand{\rsmatrix}{\end{smallmatrix}\right)}

\stackMath

\stackMath

\makeatletter
\newcommand*\rel@kern[1]{\kern#1\dimexpr\macc@kerna}
\newcommand*\widebar[1]{%
  \begingroup
  \def\mathaccent##1##2{%
    \rel@kern{0.8}%
    \overline{\rel@kern{-0.8}\macc@nucleus\rel@kern{0.2}}%
    \rel@kern{-0.2}%
  }%
  \macc@depth\@ne
  \let\math@bgroup\@empty \let\math@egroup\macc@set@skewchar
  \mathsurround\z@ \frozen@everymath{\mathgroup\macc@group\relax}%
  \macc@set@skewchar\relax
  \let\mathaccentV\macc@nested@a
  \macc@nested@a\relax111{#1}%
  \endgroup
}

\counterwithin*{equation}{part}
\counterwithin*{thm}{part}
\counterwithin*{figure}{part}

\tikzset{meter/.append style={draw, inner sep=10, rectangle, font=\vphantom{A}, minimum width=30, line width=.8, path picture={\draw[black] ([shift={(.1,.3)}]path picture bounding box.south west) to[bend left=50] ([shift={(-.1,.3)}]path picture bounding box.south east);\draw[black,-latex] ([shift={(0,.1)}]path picture bounding box.south) -- ([shift={(.3,-.1)}]path picture bounding box.north);}}}
\tikzset{roundnode/.append style={circle, draw=black, fill=gray!20, thick, minimum size=10mm}}
\tikzset{squarenode/.style={rectangle, draw=black, fill=none, thick, minimum size=10mm}}

\definecolor{Blues5seq1}{RGB}{239,243,255}
\definecolor{Blues5seq2}{RGB}{189,215,231}
\definecolor{Blues5seq3}{RGB}{107,174,214}
\definecolor{Blues5seq4}{RGB}{49,130,189}
\definecolor{Blues5seq5}{RGB}{8,81,156}

\definecolor{Greens5seq1}{RGB}{237,248,233}
\definecolor{Greens5seq2}{RGB}{186,228,179}
\definecolor{Greens5seq3}{RGB}{116,196,118}
\definecolor{Greens5seq4}{RGB}{49,163,84}
\definecolor{Greens5seq5}{RGB}{0,109,44}

\definecolor{Reds5seq1}{RGB}{254,229,217}
\definecolor{Reds5seq2}{RGB}{252,174,145}
\definecolor{Reds5seq3}{RGB}{251,106,74}
\definecolor{Reds5seq4}{RGB}{222,45,38}
\definecolor{Reds5seq5}{RGB}{165,15,21}

\allowdisplaybreaks

\newcommand{\GG}{\pazocal{G}}

\begin{document}

\title{Optimising quantum data hiding}

\author{Francesco Anna Mele}
\email{francesco.mele@sns.it}
\affiliation{NEST, Scuola Normale Superiore and Istituto Nanoscienze, Piazza dei Cavalieri 7, IT-56126 Pisa, Italy}

\author{Ludovico Lami}
\email{ludovico.lami@gmail.com}
\affiliation{Scuola Normale Superiore, Piazza dei Cavalieri 7, 56126 Pisa, Italy}
\affiliation{QuSoft, Science Park 123, 1098 XG Amsterdam, The Netherlands}
\affiliation{Korteweg--de Vries Institute for Mathematics, University of Amsterdam, Science Park 105-107, 1098 XG Amsterdam, The Netherlands}
\affiliation{Institute for Theoretical Physics, University of Amsterdam, Science Park 904, 1098 XH Amsterdam, The Netherlands}

\begin{abstract}
Quantum data hiding is the existence of pairs of bipartite quantum states that are (almost) perfectly distinguishable with global measurements, yet close to indistinguishable when only measurements implementable with local operations and classical communication are allowed. Remarkably, data hiding states can also be chosen to be separable, meaning that secrets can be hidden using no entanglement that are almost irretrievable without entanglement --- this is sometimes called `nonlocality without entanglement'. Essentially two families of data hiding states were known prior to this work: Werner states and random states. Hiding Werner states can be made either separable or globally \emph{perfectly} orthogonal, but not both --- separability comes at the price of orthogonality being only approximate. Random states can hide many more bits, but they are typically entangled and again only approximately orthogonal. In this paper, we present an explicit construction of novel group-symmetric data hiding states that are simultaneously separable, perfectly orthogonal, and even invariant under partial transpose, thus exhibiting the phenomenon of nonlocality without entanglement to the utmost extent. Our analysis leverages novel applications of numerical analysis tools to study convex optimisation problems in quantum information theory, potentially offering technical insights that extend beyond this work.


%
\end{abstract}

\maketitle
\tableofcontents

\section{Introduction}

Quantum data hiding~\cite{dh-original-1, dh-original-2} is one of the most bizarre phenomena that arise when quantum systems are used to store classical information. It refers to the existence of pairs of states on a bipartite quantum system that can be perfectly distinguished using global measurements acting jointly on both parties, yet remain nearly indistinguishable when the parties are restricted to local operations assisted by classical communication (LOCC). This makes it possible to hide a bit of information in a bipartite quantum system in such a way that it stays essentially irretrievable unless the two parties can exchange quantum information. Loosely speaking, quantum data hiding can be regarded as a quantum analogue of the classical phenomenon of secret sharing~\cite{Shamir1979}, yet it is strictly stronger, because classical communication breaks secret sharing but not quantum data hiding.

Beyond its cryptographic relevance, quantum data hiding has been suggested~\cite{private, Christandl2017} to play a key role in entanglement theory, particularly through its link with bound entanglement~\cite{Horodecki-review, Horodecki-open-problems}, a form of entanglement that cannot be distilled into ebits via LOCC, even when arbitrarily many state copies are available. The idea is as follows. Given a data hiding pair, one can construct a four-partite state, shared between two agents Alice and Bob, where one Alice-Bob pair of systems is called the \emph{shield} and the other, consisting of a single qubit per party, is called the \emph{key}. The shield can hide a bit $a\in {0,1}$ via a data hiding pair, while the key is prepared in one of the two Bell states $\ket{\Psi_a} \coloneqq \left(\ket{00} + (-1)^a \ket{11}\right)/\sqrt2$, depending on $a$. The intuition expressed in~\cite{private} is that any entanglement present in such a state should be essentially undistillable, meaning that even many copies of the state cannot be converted by LOCC into one close to a pure ebit. In quantum information parlance, the state should be \emph{bound entangled}. Indeed, on the one hand, the entanglement cannot be retrieved without knowledge of $a$, since $\Psi_0 + \Psi_1 = \ketbra{00} + \ketbra{11}$ is (proportional to) a separable state. On the other hand, the value of $a$ cannot be recovered reliably without global measurements, which are not available under LOCC. Constructing data hiding states thus provides natural candidates for bound entanglement~\cite{private,Christandl2017,Horodecki-review,Horodecki-open-problems}.

In the original works~\cite{dh-original-1, dh-original-2}, an example of a data hiding pair was provided using the two extremal Werner states~\cite{Werner}, i.e.\ the normalised projectors onto the fully antisymmetric and fully symmetric subspaces of a bipartite Hilbert space $\C^d\otimes \C^d$. These two states, hereafter called the antisymmetric and symmetric states, respectively, are orthogonal and hence perfectly distinguishable under global (projective) measurements, yet the highly nonlocal nature of their supports makes them difficult to distinguish using LOCC. A possible drawback of this construction, however, is that entanglement is required to implement it in the first place. While the symmetric state is separable, i.e.\ unentangled, the antisymmetric state does contain some entanglement~\cite{Christandl2012}. One can remedy this by constructing two \emph{separable} states that are nearly indistinguishable under LOCC, but in that case perfect orthogonality must be sacrificed~\cite{Eggeling2002}. Other randomised constructions~\cite{Aubrun-2015} preserve perfect orthogonality but typically involve highly entangled states.

Since these works, the problem of finding a separable, perfectly orthogonal data hiding pair---namely, a pair of separable quantum states with orthogonal supports that are nevertheless nearly indistinguishable under LOCC---has remained open. Here we solve this problem by providing an explicit construction based on a family of group symmetric states. Our strategy is to first construct a pair of orthogonal states that are only imperfectly data hiding, i.e.\ that can still be discriminated with \emph{some} accuracy using LOCC, and then to boost their LOCC indistinguishability by using a trick similar to that of the original work~\cite{dh-original-1}, namely, hiding a classical bit into the parity of a long string of bits encoded into the base pair.

The paper is organised as follows: Section~\ref{prob_statement} formalises the problem of quantum data hiding and states our main result; Section~\ref{sec_proof} contains its proof; Section~\ref{sec_conclusion} presents the conclusions and open questions.

\section{Problem statement}\label{prob_statement}
In this section we introduce the notation needed to formalise the concept of quantum data hiding and to mathematically state our main result.

Consider a bipartite quantum system. 
Matthews, Wehner, and Winter~\cite{VV-dh} introduced the class of \emph{LOCC POVMs}, consisting of all measurements implementable by local operations and classical communication between the two parties. 
They further defined the associated \emph{LOCC norm}, denoted $\|\cdot\|_{\mathrm{LOCC}}$, which naturally arises in the study of quantum state discrimination~\cite{KHATRI} when restricted to LOCC POVMs. 
Specifically, let $\rho_1$ and $\rho_2$ be bipartite quantum states. 
Suppose one is given a single copy of an unknown state, promised to be either $\rho_1$ or $\rho_2$ with equal prior probability. 
The optimal success probability of identifying the state using LOCC POVMs is~\cite{VV-dh}
\bb
  P_{\mathrm{succ}}^{(\mathrm{LOCC})}(\rho_1,\rho_2) 
  = \tfrac{1}{2} + \tfrac{1}{4}\|\rho_1 - \rho_2\|_{\mathrm{LOCC}}.
\ee
This relation can be regarded as a definition of the LOCC norm with a clear operational interpretation: $\tfrac{1}{2}\|\rho_1-\rho_2\|_{\mathrm{LOCC}}$ quantifies the maximum bias achievable in distinguishing $\rho_1$ from $\rho_2$ under LOCC. 
In other words, the larger the LOCC norm, the easier it is to discriminate the two states with LOCC measurements.

This is directly analogous to the operational meaning of the trace norm $\|\cdot\|_1$ in the context of quantum state discrimination under global measurements. 
Indeed, the celebrated Holevo-Helstrom theorem~\cite{HELSTROM,Holevo1976} establishes that, given a single copy of a state promised to be either $\rho_1$ or $\rho_2$ with equal prior probability, the optimal success probability when optimising over \emph{all} (global) POVMs is
\bb
  P_{\mathrm{succ}}^{(\mathrm{ALL})}(\rho_1,\rho_2) 
  = \tfrac{1}{2} + \tfrac{1}{4}\|\rho_1 - \rho_2\|_1\,,
\ee
where $\tfrac{1}{2}\|\rho_1-\rho_2\|_1$ is the trace distance between $\rho_1$ and $\rho_2$~\cite{KHATRI}.

We are now ready to introduce the notion of quantum data hiding. 
Informally, two bipartite states $\rho_1$ and $\rho_2$ form a data-hiding pair if they are perfectly distinguishable by global measurements (i.e.~they are orthogonal and thus $P_{\mathrm{succ}}^{(\mathrm{ALL})}(\rho_1,\rho_2)=1$), yet they remain nearly indistinguishable under LOCC, i.e.~$P_{\mathrm{succ}}^{(\mathrm{LOCC})}(\rho_1,\rho_2)\approx\tfrac{1}{2}$ (random guess) or equivalently
\bb
\tfrac{1}{2}\|\rho_1-\rho_2\|_{\mathrm{LOCC}}\approx 0\,.
\ee
One can formalise this concept by introducing an error parameter $\varepsilon$ as follows:
\begin{Def}[($\varepsilon$-quantum data hiding states)]
Let $\varepsilon \in (0,1)$.  
A pair of quantum states $(\rho_1,\rho_2)$ is called a pair of $\varepsilon$-quantum data hiding states if they are orthogonal and satisfy
\bb
  \frac{1}{2}\,\|\rho_1 - \rho_2\|_{\mathrm{LOCC}} \le \varepsilon.
\ee
\end{Def}
Previously, $\varepsilon$-quantum data hiding states were known to exist only when the states are entangled~\cite{Christandl2012,Aubrun-2015}. 
Additionally, separable pairs that are nearly indistinguishable under LOCC were also constructed~\cite{Eggeling2002}, but these are not perfectly orthogonal and therefore cannot be perfectly distinguished by global measurements, making them not fully satisfactory for data hiding. In this work we show that all these requirements can be satisfied simultaneously: orthogonality, separability, and $\varepsilon$-indistinguishability under LOCC. Namely, we prove that for every $\varepsilon\in(0,1)$ there exist separable $\varepsilon$-quantum data hiding states, as stated in Theorem~\ref{main_thm_data} below. Moreover, our construction is explicit and provides quantitative bounds on the required local dimension. 
\begin{thm}[(Existence of separable, orthogonal quantum data hiding states)]\label{main_thm_data}
For every $\varepsilon\in(0,1)$ there exist bipartite states $\rho_1,\rho_2$ on $\C^{D}\otimes\C^{D}$ that are both separable and orthogonal, and satisfy
\bb
  \tfrac{1}{2}\|\rho_1-\rho_2\|_{\mathrm{LOCC}}\le \varepsilon,
\ee
with local dimension bounded as $D \le 40\bigl(\tfrac{2}{\varepsilon}\bigr)^{10}$. 
\end{thm}
The explicit construction of $\rho_1$ and $\rho_2$, together with the proof of the theorem, is provided in Section~\ref{explicit_costruction} below. 
The theorem shows that a local dimension of at most $D=O(1/\varepsilon^{10})$ is sufficient to construct separable $\varepsilon$-quantum data hiding states. For vanishing $\varepsilon$, this dimension diverges; however, this is not a limitation of our construction but an inherent feature of any quantum data hiding scheme. 
In fact, one can show that a local dimension of at least $D=\Omega(1/\varepsilon)$ is required to realise $\varepsilon$-quantum data hiding states:

\begin{rem}[(Required dimension for quantum data hiding)]\label{converse_rem}
Let $\varepsilon\in(0,1)$. Assume there exists a local dimension $D\in\mathbb{N}$ such that there are orthogonal bipartite states $\rho_1,\rho_2$ on $\C^{D}\otimes\C^{D}$ satisfying $\tfrac{1}{2}\|\rho_1-\rho_2\|_{\mathrm{LOCC}}\le \varepsilon$. Then necessarily $D\ge \frac{1}{2}+\tfrac{1}{2\varepsilon}$.

This follows directly from \cite[Eq.~(43)]{ultimate} , which lower bounds the LOCC norm in terms of the trace norm as $\|\cdot\|_{\mathrm{LOCC}} \;\ge\; \tfrac{1}{2D-1}\,\|\cdot\|_1$ (see also~\cite[Corollary~17]{VV-dh} for a weaker bound). Applying this to $\rho_1-\rho_2$ gives
\bb
  \varepsilon \;\ge\; \tfrac{1}{2}\|\rho_1-\rho_2\|_{\mathrm{LOCC}}
  \;\ge\; \tfrac{1}{2D-1}\,\tfrac{1}{2}\|\rho_1-\rho_2\|_1
  \;=\; \tfrac{1}{2D-1}\,,
\ee
where the last equality uses that $\rho_1$ and $\rho_2$ are orthogonal, thus proving the claim.
\end{rem}

\section{Construction of separable, orthogonal quantum data hiding states}\label{sec_proof}

Our construction of states that are nearly indistinguishable under LOCC follows the strategy introduced in the foundational works on quantum data hiding~\cite{dh-original-1,dh-original-2}. 
We begin with a pair of bipartite states $\sigma_0,\sigma_1$ on $\mathbb{C}^d\otimes \mathbb{C}^d$ that are only imperfectly distinguishable under LOCC, i.e.~$\tfrac{1}{2}\|\sigma_1-\sigma_0\|_{\mathrm{LOCC}}<1$. 
From such a pair, one can try to generate new states that are harder to distinguish under LOCC by encoding parity information. 
Specifically, for any $k\in\mathbb{N}$ we define the \emph{odd state} $\rho_1^{(k)}$ and the \emph{even state} $\rho_0^{(k)}$ on $\bigl(\mathbb{C}^d\bigr)^{\otimes k}\otimes \bigl(\mathbb{C}^d\bigr)^{\otimes k}$ as
\bb\label{eq_evenodd}
  \rho_1^{(k)} &\coloneqq \frac{1}{2^{k-1}}
  \sum_{\substack{x_1,\ldots,x_k\in\{0,1\}\\ x_1+\cdots+x_k\equiv 1\ (\mathrm{mod}\,2)}}
  \sigma_{x_1}\otimes\cdots\otimes\sigma_{x_k}, \\
  \rho_0^{(k)} &\coloneqq \frac{1}{2^{k-1}}
  \sum_{\substack{x_1,\ldots,x_k\in\{0,1\}\\ x_1+\cdots+x_k\equiv 0\ (\mathrm{mod}\,2)}}
  \sigma_{x_1}\otimes\cdots\otimes\sigma_{x_k}.
\ee
Thus, $\rho_1^{(k)}$ (resp.~$\rho_0^{(k)}$) is the uniform mixture of tensor products $\sigma_{x_1}\otimes\cdots\otimes\sigma_{x_k}$ over all odd (resp.~even) parity strings. 
Distinguishing $\rho_1^{(k)}$ from $\rho_0^{(k)}$ is therefore equivalent to determining the parity of the number of copies of $\sigma_1$ in the mixture of $k$ quantum systems. 
Intuitively, this task should become increasingly difficult as $k$ grows. 
Formally, one may conjecture that for all $\sigma_0,\sigma_1$ with $\tfrac{1}{2}\|\sigma_1-\sigma_0\|_{\mathrm{LOCC}}<1$, the corresponding even and odd states satisfy
\bb
    \lim_{k\to\infty}\tfrac{1}{2}\|\rho_1^{(k)}-\rho_0^{(k)}\|_{\mathrm{LOCC}} \stackrel{?}{=} 0\,.
\ee
Since it holds that $\frac{\rho_1^{(k)}-\rho_0^{(k)}}{2}=\left(\frac{\sigma_1-\sigma_0}{2}\right)^{\otimes k}$, as is easily verified, the conjecture can be restated as follows:
\begin{cj}\label{cj1}
    For all bipartite states $\sigma_0,\sigma_1$ with $\tfrac{1}{2}\|\sigma_1-\sigma_0\|_{\mathrm{LOCC}}<1$, it holds that
    \bb 
       \lim_{k\to\infty}\left\|\left(\frac{\sigma_1-\sigma_0}{2}\right)^{\otimes k}\right\|_{\mathrm{LOCC}}=0\,.
    \ee
\end{cj}

A proof of Conjecture~\ref{cj1} would yield many examples of separable, orthogonal quantum data hiding states. 
Indeed, if $\sigma_0$ and $\sigma_1$ are also separable and orthogonal, then for every $k$ the associated states $\rho^{(k)}_0$ and $\rho^{(k)}_1$ remain separable and orthogonal. 
Thus, Conjecture~\ref{cj1} would directly imply:

\begin{cj}[(Construction of orthogonal, separable quantum data hiding states)]\label{cj2}
    Let $\sigma_0,\sigma_1$ be bipartite states that are orthogonal, separable, and satisfy $\tfrac{1}{2}\|\sigma_1-\sigma_0\|_{\mathrm{LOCC}}<1$. 
    Then the associated even state $\rho^{(k)}_0$ and odd state $\rho^{(k)}_1$, defined in~\eqref{eq_evenodd}, satisfy
    \bb\label{eq_conjectured}
       \lim_{k\to\infty}\tfrac{1}{2}\left\|\rho_1^{(k)}-\rho_0^{(k)}\right\|_{\mathrm{LOCC}} = 0\,.
    \ee
    Equivalently, for all $\varepsilon\in(0,1)$, the even and odd states $\rho^{(k)}_0,\rho^{(k)}_1$ form a pair of separable, orthogonal $\varepsilon$-quantum data hiding states for sufficiently large $k$.
\end{cj}

While we are not able to prove Conjecture~\ref{cj2} in full generality, we construct explicit families of separable, orthogonal states $\sigma_0,\sigma_1$ for which~\eqref{eq_conjectured} holds, thereby establishing our main result stated in Theorem~\ref{main_thm_data}: the existence of separable, orthogonal quantum data hiding states. The construction of such states $\sigma_0,\sigma_1$ is provided in the forthcoming subsection.

\subsection{Two special states}\label{explicit_costruction}
Consider a bipartite system with Hilbert space $\C^d\otimes \C^d$.  
Define the operators
\bb\label{def_theta}
  \Theta_0 \coloneqq \Phi,\qquad 
  \Theta_1 \coloneqq P - \Phi,\qquad  
  \Theta_2 \coloneqq Q_+,\qquad  
  \Theta_3 \coloneqq Q_-,
\ee
where
\bb\label{def_orth_proj}
  \Phi &\coloneqq \frac1d \sum_{i,j=0}^{d-1} \ketbraa{i}{j}\otimes \ketbraa{i}{j},\\ 
  P &\coloneqq \sum_{i=0}^{d-1} \ketbra{i}\otimes \ketbra{i}, \\
  Q_+ &\coloneqq \tfrac{\id + F - 2P}{2}, \\
  Q_- &\coloneqq \tfrac{\id - F}{2}, \\
  F &\coloneqq \sum_{i,j=0}^{d-1}\ketbraa{i}{j}\otimes\ketbraa{j}{i}.
\ee
Here $F$ is the flip (swap) operator, $\Phi$ the maximally entangled state, $P$ the projector onto the maximally correlated subspace, and $Q_-$ the projector onto the antisymmetric subspace.  
It is straightforward to check that $\Theta_0,\Theta_1,\Theta_2,\Theta_3$ are mutually orthogonal projectors with ranks
\bb\label{tr_theta}
  \Tr\Theta_0 = 1,\qquad 
  \Tr\Theta_1 = d-1,\qquad  
  \Tr\Theta_2 = \tfrac{d(d-1)}{2},\qquad  
  \Tr\Theta_3 = \tfrac{d(d-1)}{2},
\ee
and they resolve the identity, i.e.~$\sum_{i=0}^3 \Theta_i = \id$. 

The \emph{antisymmetric state}~\cite{Christandl2012}, which can be defined as $\alpha\coloneqq \frac{\Theta_3}{\Tr \Theta_3}$, is universally regarded as one of the best candidates for counterexamples in quantum information theory~\cite{Christandl2012}. However, more recently another state has been claiming the throne~\cite{irreversibility}: the state $\omega\coloneqq \frac{\Theta_1}{\Tr \Theta_1}$, which is the normalised projector onto the $(d-1)$-dimensional subspace orthogonal to the maximally entangled state within the maximally correlated subspace. Now the forbidden question is: \emph{what happens if one mixes them?} Following this somehow outrageous idea, let us look at the state
\bb
\sigma_1^{(d)} \coloneqq&\ \frac12\left(\alpha + \omega\right)=\tfrac{1}{2(d-1)}\,\Theta_1+\tfrac{1}{d(d-1)}\,\Theta_3\,,
\ee
which might be called the \emph{biblical state}, for it mixes the alpha and the omega. An orthogonal state that nicely pairs up with this one is
\bb
\sigma_0^{(d)} \coloneqq \frac{1}{d}\Phi+\left(1-\frac{1}{d}\right)\frac{Q_-}{\Tr Q_-}=\tfrac{1}{d}\,\Theta_0+\tfrac{2}{d^2}\,\Theta_2\,.
\ee
Our construction is based precisely on these two states, which we summarise in the following definition for ease of reference.

\begin{Def}[(Two special states)]\label{def_sigma01}
Let $\sigma_0^{(d)},\sigma_1^{(d)}$ be two states on $\C^d\otimes\C^d$ defined as
\bb\label{sigma01}
  \sigma_0^{(d)} \coloneqq \tfrac{1}{d}\,\Theta_0+\tfrac{2}{d^2}\,\Theta_2,
  \qquad
  \sigma_1^{(d)} \coloneqq \tfrac{1}{2(d-1)}\,\Theta_1+\tfrac{1}{d(d-1)}\,\Theta_3.
\ee
\end{Def}
By construction, $\sigma_0^{(d)}$ and $\sigma_1^{(d)}$ are valid quantum states and they are orthogonal. In Appendix~\ref{sec_appendix_c}, we show that they are invariant under partial transposition, and hence both are PPT states. Moreover, results from~\cite{park2023universal} imply that these states are not only PPT but in fact separable. For completeness, we present an independent proof of this result below.

\begin{lemma}\label{lemma_sepsep}
The states $\sigma_0^{(d)}$ and $\sigma_1^{(d)}$ in~\eqref{sigma01} are orthogonal and separable. 
\end{lemma}
Before proving the lemma, let us establish a useful tool. Let us define the \emph{$\GG$-twirling channel} 
\bb\label{G_twirling} \TT_\GG (X) \coloneqq \frac{1}{|\GG|} \sum_{U\in \GG} (U\otimes U)\, X \,(U\otimes U)^\dag, 
\ee 
where $\GG$ is the group of $d\times d$ unitaries \bb \GG \coloneqq \{ U_\pi V_\e:\ \pi \in S_d,\ \e\in \{- 1,1\}^d \}, \ee with $U_\pi \coloneqq \sum_{i=0}^{d-1} \ketbraa{\pi(i)}{i}$ implementing the permutation $\pi\in S_d$, and $V_\e \coloneqq \sum_{i=0}^{d-1} \e_i \ketbra{i}$ a diagonal Hermitian unitary. \begin{lemma}\label{lemma_G_twirling_simplified}
For all operators $X$, the $\GG$-twirling acts as \bb\label{G_twirling_simplified} \TT_\GG(X) = \sum_{i=0}^3 \frac{\Tr[X\,\Theta_i]}{\Tr\Theta_i}\, \Theta_i, \ee where $\Theta_0,\Theta_1,\Theta_2,\Theta_3$ are the four mutually orthogonal projectors in~\eqref{def_theta}. \end{lemma} A proof is given in Appendix~\ref{sec_appendix_a}. Note also that $\TT_\GG$ is an LOCC channel, as it can be implemented via the following LOCC protocol: (i) Alice samples $U\in\GG$ uniformly at random; (ii) she communicates which $U$ has been sampled to Bob via classical communication; (iii) both parties apply $U$ locally. We can now prove Lemma~\ref{lemma_sepsep}. \begin{proof}[Proof of Lemma~\ref{lemma_sepsep}] By exploiting Lemma~\ref{lemma_G_twirling_simplified}, a direct calculation shows that \bb\label{output_twirling_xi_pm} \sigma_0^{(d)} &= \TT_\GG\!\left(\ketbra{e}\otimes \ketbra{e}\right), \\ \sigma_1^{(d)} &= \TT_\GG\!\left(\ketbra{+}\otimes \ketbra{-}\right), \ee where $\ket{e} \coloneqq \tfrac{1}{\sqrt{d}} \sum_{i=0}^{d-1} \ket{i}$, and $\ket{\pm} \coloneqq \tfrac{1}{\sqrt{2}} \bigl(\ket{0}\pm \ket{1}\bigr)$. This demonstrates that $\sigma_0^{(d)}$ and $\sigma_1^{(d)}$ can be obtained as the outputs of $\TT_\GG$ acting on product states. Since $\TT_\GG$ is an LOCC channel, it follows that $\sigma_0^{(d)}$ and $\sigma_1^{(d)}$ are separable. This establishes the claim. \end{proof}

We also quantify how well $\sigma_0^{(d)}$ and $\sigma_1^{(d)}$ can be distinguished by LOCC.

\begin{prop}[(Bounds on the LOCC norm between the two special states)]\label{locc_norm_xi}
For all $d\ge 2$,
\bb
  \frac12 - \frac{1}{d} \;\le\; \frac12\left\| \sigma_0^{(d)} - \sigma_1^{(d)} \right\|_\locc \;\le\; \frac12 + \frac1d\, .
\ee
In particular, for $d\ge 3$ we have $\frac12\left\| \sigma_0^{(d)} - \sigma_1^{(d)}\right\|_\locc<1$.
\end{prop}

\noindent The proof is given in Appendix~\ref{sec_appendix_b}.  As a consequence of Proposition~\ref{locc_norm_xi}, for $d\ge 3$ the optimal LOCC protocol to distinguish the equiprobable $\sigma_1^{(d)}$ and $\sigma_0^{(d)}$ succeeds with probability strictly smaller than one. We are therefore in the setting discussed above: $\sigma_0^{(d)}$ and $\sigma_1^{(d)}$ are orthogonal, separable, and only imperfectly distinguishable under LOCC.  
Following Conjecture~\ref{cj2}, we now amplify indistinguishability via the parity construction. That is, for $k,d\in\N$, we define the odd and even state on $\bigl(\C^d\bigr)^{\otimes k}\otimes \bigl(\C^d\bigr)^{\otimes k}$ as:
\bb\label{eq_evenodd2}
\begin{aligned}
  \rho_1^{(k,d)} &\coloneqq \frac{1}{2^{k-1}}
    \sum_{\substack{x_1,\ldots,x_k\in\{0,1\}\\ x_1+\cdots+x_k\equiv 1\ (\mathrm{mod}\,2)}}
    \sigma_{x_1}^{(d)}\otimes\cdots\otimes\sigma_{x_k}^{(d)}, \\
  \rho_0^{(k,d)} &\coloneqq \frac{1}{2^{k-1}}
    \sum_{\substack{x_1,\ldots,x_k\in\{0,1\}\\ x_1+\cdots+x_k\equiv 0\ (\mathrm{mod}\,2)}}
    \sigma_{x_1}^{(d)}\otimes\cdots\otimes\sigma_{x_k}^{(d)}.
\end{aligned}
\ee
Our main technical contribution is an upper bound on the LOCC norm between the even and odd states. This is given in the following proposition, which forms the core of our analysis.

\begin{prop}[(Upper bound on the LOCC norm between even and odd states)]\label{main_PROP}
Let $d,k\in\N$ with $d\ge 2$. Then
\bb
  \tfrac12 \bigl\|\rho_1^{(k,d)}-\rho_0^{(k,d)}\bigr\|_{\locc}\;\le\; 2\,\mu_d^k,
\ee
where the quantity $\mu_d$ (plotted in Fig.~\ref{fig:mu_d}) is defined as
\bb\label{def_mud}
    \mu_d =\sqrt{\,1- \frac{ \frac{5}{8}+\frac{1}{d}\left(\tfrac{1}{4}+\tfrac{2}{d}+\tfrac{9}{d^{2}}-\tfrac{6}{d^{3}} -\;\sqrt{2}\,\Bigl(\tfrac{9}{4}+\tfrac{3}{d}+\tfrac{1}{d^{2}}\Bigr) \sqrt{1-\frac{2}{d+\frac{4}{d}}}\right) }{ 1+\tfrac{2}{d}+\tfrac{4}{d^{2}} }\,}\;.
\ee
\end{prop}

\begin{figure}[h!]
  \centering
  \includegraphics[width=0.8\linewidth]{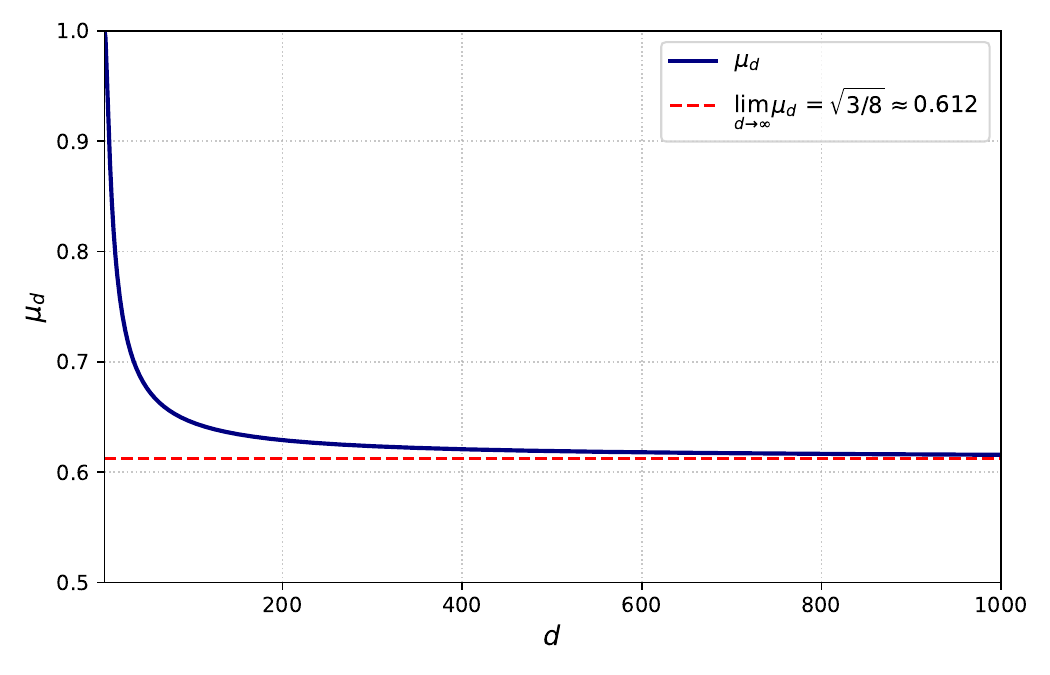}
  \caption{Behaviour of $\mu_d$ for $2 \le d \le 1000$. The function is monotonically decreasing in the parameter $d$, 
with $\mu_2 = 1$, $\mu_3 \approx 0.993$, and asymptotic value 
$\lim_{d \to \infty} \mu_d = \sqrt{3/8} \approx 0.612$. Crucially, it satisfies $\mu_d<1$ for all $d\ge3$.}
  \label{fig:mu_d}
\end{figure}

\noindent The proof is deferred to the Subsection~\ref{sec_proof_of_prop_main}. The behaviour of $\mu_d$ is shown in Fig.~\ref{fig:mu_d}: it decreases monotonically with $d$ and satisfies $\mu_d<1$ for all $d\ge3$. Consequently, Proposition~\ref{main_PROP} implies that
\bb
  \lim_{k\to\infty}\,\tfrac12 \bigl\|\rho_1^{(k,d)}-\rho_0^{(k,d)}\bigr\|_{\locc}=0,
\ee
so the odd and even states become asymptotically indistinguishable under LOCC, while remaining separable and orthogonal. They thus provide examples of  separable and orthogonal quantum data hiding states. We are therefore ready to prove our main result, Theorem~\ref{main_thm_data}.
\begin{proof}[Proof of Theorem~\ref{main_thm_data}]
    Fix $\varepsilon\in(0,1)$. By Proposition~\ref{main_PROP}, whenever $k$ is chosen large enough that $2\mu_d^k\le \varepsilon$, the states $\rho_0^{(k,d)}$ and $\rho_1^{(k,d)}$ form a pair of (separable, orthogonal) $\varepsilon$-quantum data hiding states. Since their associated local dimension is $d^k$, it follows that for any $\varepsilon\in(0,1)$ we can construct separable, orthogonal $\varepsilon$-quantum data hiding states with local dimension
\bb
  D_\varepsilon \;\coloneqq\;
  \min_{\substack{d\in\N,\ d\ge 2\\ k\in\N\\ 2\mu_d^k\le\varepsilon}} d^k.
\ee
Since for fixed $d$ the smallest admissible $k$ is $k=\Bigl\lceil \tfrac{\log(2/\varepsilon)}{\log(1/\mu_d)} \Bigr\rceil$, we can equivalently express 
\bb 
D_\varepsilon = \min_{d\in\N,\ d\ge 2}\; d^{\left\lceil \tfrac{\log(2/\varepsilon)}{\log(1/\mu_d)} \right\rceil}. 
\ee 
A numerical search reveals that the optimum is attained at $d=40$. Substituting this value yields 
\bb 
D_\varepsilon \;\le\; 40^{\left\lceil \frac{\log(2/\varepsilon)}{\log\left(1/\mu_{40}\right) } \right\rceil} \;\le\; 40\left(\frac{2}{\varepsilon}\right)^{\frac{\log 40}{\log\left(1/\mu_{40}\right)}} \;\le\; 40\left(\frac{2}{\varepsilon}\right)^{10}. 
\ee
Hence a local dimension of $D_\varepsilon\le40\,(2/\varepsilon)^{10}$ suffices to construct separable, orthogonal $\varepsilon$-quantum data hiding states. This concludes the proof.
\end{proof}

\subsection{Bounding the LOCC norm via the PPT norm}
To prove Proposition~\ref{main_PROP}, we have to upper bound the LOCC norm between the even and odd states.  A common approach in entanglement theory to deal with an optimisation over LOCC protocols is to relax it to the larger, more tractable class of \emph{PPT} protocols~\cite{Horodecki-review}.  
In this spirit, we will upper bound the LOCC norm by the \emph{PPT norm}~\cite{VV-dh}, denoted as $\|\cdot\|_{\ppt}$ and defined as follows.

Consider two bipartite states $\rho_1,\rho_2$, and suppose we are given a single copy of an unknown state, promised to be either $\rho_1$ or $\rho_2$ with equal prior probability. The optimal success probability of correctly identifying the state using PPT measurements is~\cite{VV-dh}
\begin{align}
P_{\mathrm{succ}}^{(\mathrm{PPT})}(\rho_1,\rho_2) 
  \coloneqq \max_{(M_1,M_2)\in \mathrm{PPT\;POVM}}
     \Bigl(\tfrac12\Tr[M_1\rho_1]+\tfrac12\Tr[M_2\rho_2]\Bigr),
\end{align}
where the maximisation is over the set of PPT POVMs~\cite{VV-dh},
\begin{align}
\mathrm{PPT\;POVM}
  \coloneqq \bigl\{(M_1,M_2):\; M_1,M_2\ge0,\; M_1+M_2=\id,\;
  M_1^\Gamma\ge 0,\; M_2^\Gamma\ge 0 \bigr\},
\end{align}
and $X^\Gamma$ denotes the partial transpose of $X$. This expression can be rewritten as
\begin{align}
P_{\mathrm{succ}}^{(\mathrm{PPT})}(\rho_1,\rho_2)
 &= \tfrac12 + \tfrac12\max_{(M_1,M_2)\in \mathrm{PPT\;POVM}} \Tr\!\bigl[M_1(\rho_1-\rho_2)\bigr] \\
 &= \tfrac12 + \tfrac14 \,\|\rho_1-\rho_2\|_{\ppt},
\end{align}
which defines the PPT norm via
\begin{align}
\tfrac12\,\|\rho_1-\rho_2\|_{\ppt}
 \;\coloneqq\; \max_{(M_1,M_2)\in \mathrm{PPT\;POVM}} \Tr\!\bigl[M_1(\rho_1-\rho_2)\bigr].
\end{align}
Equivalently, one can easily prove that this optimisation can also be expressed as
\begin{align}\label{eq_equiv_ppt}
\tfrac12\,\|\rho_1-\rho_2\|_{\ppt}
 = \max_{\substack{0\le M\le \id\\ 0\le M^\Gamma\le \id}} \Tr\!\bigl[M(\rho_1-\rho_2)\bigr]\,.
\end{align}
As a concrete example, in Appendix~\ref{sec_appendix_b} we show that for the states $\sigma_0^{(d)},\sigma_1^{(d)}$ defined in Definition~\ref{def_sigma01}, the PPT norm can be evaluated in closed form, yielding
\bb
  \tfrac12\left\|\sigma_0^{(d)}-\sigma_1^{(d)}\right\|_{\ppt} \;=\; \tfrac12 + \tfrac{1}{d}, 
  \qquad \forall\, d\ge 2\,.
\ee
Since every LOCC measurement is also a PPT measurement, it follows that
\bb
  \|\cdot\|_{\locc} \;\le\; \|\cdot\|_{\ppt}.
\ee
This key observation allows us to control the LOCC norm by bounding instead the more tractable PPT norm, which is precisely the strategy we shall follow in the proof of Proposition~\ref{main_PROP} in Subsection~\ref{sec_proof_of_prop_main}.

\subsection{Bounding the PPT norm}
In the following lemma we prove that the PPT norm between the even and odd states can be written in terms of a simplified optimisation problem.
\begin{lemma}[(Simplified optimisation for the PPT norm between even and odd states)]\label{lemma_lp}
For all $d,k\in\N$ with $d\ge 2$, the PPT norm between the even and odd states can be expressed as
\bb\label{linear_program}
\tfrac12 \bigl\|\rho_1^{(k,d)}-\rho_0^{(k,d)}\bigr\|_{\ppt}
  \;=\; \inf_{x\in\R^{4^k}} \Bigl( \|x\|_1 \;+\; \|\bar{r}_d^{\otimes k}-W_d^{\otimes k}x\|_1 \Bigr),
\ee
where
\bb\label{rdwd}
  \bar{r}_d \;\coloneqq\;
  \begin{pmatrix}
    \tfrac{1}{2d} \\[0.6ex]
    -\tfrac{1}{4} \\[0.6ex]
    \tfrac{d-1}{2d} \\[0.6ex]
    -\tfrac{1}{4}
  \end{pmatrix},
\qquad
W_d \;\coloneqq\;
\begin{pmatrix}
  \tfrac{1}{d} & \tfrac{1}{d} & \tfrac{1}{d} & -\tfrac{1}{d} \\[1.0ex]
  1-\tfrac{1}{d} & 1-\tfrac{1}{d} & -\tfrac{1}{d} & \tfrac{1}{d} \\[1.0ex]
  \tfrac{d-1}{2} & -\tfrac{1}{2} & \tfrac{1}{2} & \tfrac{1}{2} \\[1.0ex]
  -\tfrac{d-1}{2} & \tfrac{1}{2} & \tfrac{1}{2} & \tfrac{1}{2}
\end{pmatrix}.
\ee
Here, $\|x\|_1\coloneqq \sum_i |x_i|$ denotes the $\ell_1$ norm of $x$.
\end{lemma}
\begin{proof} 
 Let us start by observing that
    \bb\label{max_problem}
        \frac{1}{2}\|\rho_0^{(k,d)}-\rho_1^{(k,d)}\|_{\ppt}\,\eqt{(i)}\  &\max_{\substack{0\le E\le \mathbb{1} \\ 0\le E^\Gamma\le \mathbb{1}}}\Tr\left[E\,\big(\rho_0^{(k,d)}-\rho_1^{(k,d)}\big)\right]\\
\eqt{(ii)}\ &\frac{1}{2^{k-1}}
        \max_{\substack{0\le E\le \mathbb{1} \\ 0\le E^\Gamma\le \mathbb{1}}} \Tr\left[E\,(\sigma_0^{(d)}-\sigma_1^{(d)})^{\otimes k}\right]\\
\eqt{(iii)}\ &\frac{1}{2^{k-1}}\max_{\substack{0\le E\le \mathbb{1} \\ 0\le E^\Gamma\le \mathbb{1}}}\Tr\left[E\,\TT_\GG^{\otimes k}\left((\sigma_0^{(d)}-\sigma_1^{(d)})^{\otimes k}\right)\right]\\
\eqt{(iv)}\ &
        \frac{1}{2^{k-1}}\max_{\substack{0\le E\le \mathbb{1} \\ 0\le E^\Gamma\le \mathbb{1}}} \Tr\left[\TT_\GG^{\otimes k}(E)\,\,(\sigma_0^{(d)}-\sigma_1^{(d)})^{\otimes k}\right]\,.
    \ee
    Here, in (i), we exploited the equivalent definition of PPT norm in \eqref{eq_equiv_ppt}. In (ii), we used that the even and odd states satisfy
    \bb
        \frac{\rho_0^{(k,d)}-\rho_1^{(k,d)}}{2}=\left(\frac{\sigma_0^{(d)}-\sigma_1^{(d)}}{2}\right)^{\otimes k}\,,
    \ee
    as it can be shown via simple algebra. In (iii), we used that the $\GG$-twirling $\TT_\GG$, defined in~\eqref{G_twirling}, satisfies $\TT_\GG(\sigma_0^{(d)})=\sigma_0^{(d)}$ and $\TT_\GG(\sigma_1^{(d)})=\sigma_1^{(d)}$. The latter identities can be easily proved exploiting the definition of $\sigma_0^{(d)}$ and $\sigma_1^{(d)}$ in~\eqref{sigma01} together with Lemma~\ref{lemma_sepsep}, which establishes that $\TT_\GG(\cdot) = \sum_{i=0}^3 \frac{\Tr[(\cdot)\,\Theta_i]}{\Tr\Theta_i}\, \Theta_i$, where $\Theta_0,\Theta_1,\Theta_2,\Theta_3$ are the four mutually orthogonal projectors defined in~\eqref{def_theta}. In (iv), we exploited the definition of $\TT_\GG$ in \eqref{G_twirling_simplified}, along with the cyclicity of the trace and the fact that summing over $U\in\GG$ is equivalent to summing over $U^\dagger \in\GG$.

    Now, note that if $E$ is an optimal solution of the maximum problem in \eqref{max_problem}, then the twirled operator $\TT_\GG^{\otimes k}(E)$ is also an optimal solution. Indeed, if $0\le E\le \mathbb{1}$ and $0\le E^\Gamma\le \mathbb{1}$, then it holds that $0\le \TT_\GG^{\otimes k}(E)\le \mathbb{1}$ and $ 0\le\TT_\GG^{\otimes k}(E^\Gamma)\le \mathbb{1}$, 
    as a consequence of the fact that $\TT_\GG$ is a positive linear superoperator. Moreover, since $U^*=U$ for all $U\in\GG$, it follows that $\TT_\GG^{\otimes k}(E^\Gamma)=\left(\TT_\GG^{\otimes k}(E)\right)^\Gamma$. Consequently, we conclude that we can restrict the maximisation in \eqref{max_problem} over operators of the form $\TT_\GG^{\otimes k}(E)$ satisfying the constaints $0\le \TT_\GG^{\otimes k}(E)\le \mathbb{1}$ and $0\le \left(\TT_\GG^{\otimes k}(E)\right)^\Gamma\le \mathbb{1}$.

    Moreover, \eqref{G_twirling_simplified} implies that $\TT_\GG^{\otimes k}(E)$ can be written as 
\bb\label{eq_ttggk}
    \TT_\GG^{\otimes k }(E)=\sum_{i_1,i_2,\ldots, i_k=0}^3 c_{i_1,i_2,\ldots, i_k}\, \Theta_{i_1}\otimes\Theta_{i_2}\ldots\otimes \Theta_{i_k}\,,
\ee
where $(c_{i_1,i_2,\ldots, i_k})_{i_1,i_2,\ldots,i_k}$ are a suitable real numbers which depends on $E$.
Since $(\Theta_i)_{i=0,1,2,3}$ are orthogonal projectors, the condition $0\le \TT_\GG^{\otimes k }(E)\le \mathbb{1}$ is equivalent to the condition 
\bb
    c_{i_1,i_2,\ldots, i_k}\in[0,1]\quad \forall\,i_1,i_2,\ldots,i_k\in\{0,1,2,3\}\,,
\ee
which can be rewritten more concisely as $c\in[0,1]^{4^k}$. In addition, a direct calculation allows one to express the partial transpose of each projector $\Theta_i$ as
\bb\label{Theta_partial_transpose}
    \Theta_{i}^\Gamma=\sum_{j=0}^3 (W_d)_{ij}\Theta_{j}\quad\forall\,i\in\{0,1,2,3\}\,,
\ee
where $W_d$ is the matrix defined in \eqref{rdwd}. Hence, combining \eqref{eq_ttggk} and \eqref{Theta_partial_transpose}, we obtain
\bb
    \left(\TT_\GG^{\otimes k }(E)\right)^\Gamma=\sum_{j_1,j_2,\ldots, j_k=0}^3\left(\sum_{i_1,i_2,\ldots, i_k=0}^3 (W_d)_{i_1j_1}(W_d)_{i_2j_2}\ldots (W_d)_{i_kj_k}\,c_{i_1,i_2,\ldots, i_k} \right)\Theta_{j_1}\otimes\Theta_{j_2}\ldots\otimes \Theta_{j_k}\,.
\ee
Consequently, the condition $0\le \left(\TT_\GG^{\otimes k }(E)\right)^\Gamma\le \mathbb{1}$ is equivalent to 
\bb
    \sum_{j_1,j_2,\ldots, j_k=0}^3 (W_d)_{j_1i_1}(W_d)_{j_2i_2}\ldots (W_d)_{j_ki_k}\,c_{j_1,j_2,\ldots, j_k}\in[0,1]\quad\forall\,i_1,i_2,\ldots,i_k\in\{0,1,2,3\}\,,
\ee
which can be concisely rewritten as $(W_d^\intercal)^{\otimes k} c\in[0,1]^{4^k}$. Moreover, by exploiting the orthogonality of the projectors $(\Theta_i)_{i=0,1,2,3}$, the expressions of the trace of these projectors provided in \eqref{tr_theta}, and the fact that
\bb
    \sigma_0^{(d)}-\sigma_1^{(d)}=\frac{1}{d}\Theta_0-\frac{1}{2(d-1)}\Theta_1+\frac{2}{d^2}\Theta_2-\frac{1}{d(d-1)}\Theta_3\,,
\ee 
we can rewrite the objective function of the maximisation problem in \eqref{max_problem} as
\bb\label{obj_function}
    \Tr\left[\TT_\GG^{\otimes k}(E)\,\,(\sigma_0^{(d)}-\sigma_1^{(d)})^{\otimes k}\right]&=\sum_{i_1,i_2,\ldots,i_k=0}^3c_{i_1,i_2,\ldots,i_k}\,(r_{d})_{i_1}(r_{d})_{i_2}\ldots (r_{d})_{i_k}\\&=\sum_{i=0}^{4^{k}-1}c_i\,(r_d^{\otimes k})_i\\&=c^\intercal r_d^{\otimes k}\,,
\ee
where we defined the vector $r_{d}\coloneqq \begin{pmatrix}
\frac{1}{d}, -\frac{1}{2}, 1-\frac{1}{d}, -\frac{1}{2}
\end{pmatrix}^\intercal$, and we used the notation $c_i=c_{i_1,i_2,\ldots,i_k}$, 
with $i\in\{0,1,\ldots,4^k-1\}$ and $i_1,i_2,\ldots,i_k\in\{0,1,2,3\}$ being related by the base-4 representation as $i= i_1\, 4^{k-1} + i_2\, 4^{k-2} + \ldots + i_{k-1}\, 4 + i_k$. Consequently, we have the PPT norm of $\rho_0^{(k,d)}-\rho_1^{(k,d)}$ can be expressed as the following linear program~\cite{KHATRI,Skrzypczyk_2023}:
\bb\label{primal_problem}
\frac{1}{2}\left\|\rho_0^{(k,d)}-\rho_1^{(k,d)}\right\|_{\ppt}= \frac{1}{2^{k-1}}&\max_{\substack{c\in[0,1]^{4^k} \\ (W_d^\intercal)^{\otimes k}\,c\in[0,1]^{4^k}}}\,  c^\intercal r_d^{\otimes k} \,.
\ee
Note that the point $c\coloneqq \frac{1}{2} \left((1,1,1,1)^\intercal\right)^{\otimes k}$ is strictly feasible, indeed $(W_d^\intercal)^{\otimes k}\,c=c\in(0,1)^{4^k}$. As a result, the linear program in \eqref{primal_problem} satisfies the Slater's condition~\cite{KHATRI}, which implies that the value of the program in \eqref{primal_problem} is equal to the value of the corresponding dual program. The latter can be found via standard methods (see e.g.~\cite{KHATRI,Skrzypczyk_2023}), and it reads:
\bb\label{dual_prog}
\frac{1}{2}\left\|\rho_0^{(k,d)}-\rho_1^{(k,d)}\right\|_{\ppt}
=\frac{1}{2^{k-1}} \inf_{\substack{ x,y,z\in\R^{4^k}_+\\  y\ge r_{d}^{\otimes k} -W_d^{\otimes k}(z-x)}} \sum_{i=0}^{4^k-1} \left[z_{i}+y_{i}\right] \,.
\ee
In a more compact form, we can write:
\bb
\frac{1}{2}\left\|\rho_0^{(k,d)}-\rho_1^{(k,d)}\right\|_{\ppt}
=\frac{1}{2^{k-1}} \inf_{\substack{ x,y,z\in\R^{4^k}_+\\  y\ge r_{d}^{\otimes k} -W_d^{\otimes k}(z-x)}} S(z+y) \,,
\ee
where we introduced the notation $S(x)$ to denote the sum of the elements of a vector $x$. For the rest of the proof, given a vector $x\in\R^{4^k}$, we will denote as $x_+$ its positive part and as $x_-$ its negative part, defined as follows:
\bb
    (x_+)_i&\coloneqq \max(0,x_i)\,,\\
    (x_-)_i&\coloneqq\max(0,-x_i)\,,
\ee
so that $x=x_+-x_-$. With this notation at hand, note that 
\bb
     \frac{1}{2}\left\|\rho_0^{(k,d)}-\rho_1^{(k,d)}\right\|_{\ppt}&\eqt{(i)}\frac{1}{2^{k-1}} \inf_{x,z\in\R^{4^k}_+} S\!\left(z+\left(r_{d}^{\otimes k}  -W_d^{\otimes k}(z-x)\right)_+\right)\\
     &\eqt{(ii)}  \frac{1}{2^{k-1}}\inf_{y\in\R^{4^k}} S\!\left(y_++\left(r_{d}^{\otimes k} -W_d^{\otimes k}y \right)_+\right)\\
    &\eqt{(iii)} 2\inf_{y\in\R^{4^k}} S\!\left(y_++\left(\bar{r}_{d}^{\otimes k} -W_d^{\otimes k}y \right)_+\right)\\
    &\eqt{(iv)} \inf_{y\in\R^{4^k}} S\!\left(y+|y|+\bar{r}_{d}^{\otimes k} -W_d^{\otimes k}y +\left|\bar{r}_{d}^{\otimes k} -W_d^{\otimes k}y \right|\right)\\
    &\eqt{(v)} \inf_{y\in\R^{4^k}} \left[ S(y)+\|y\|_1+S(\bar{r}_{d}^{\otimes k}) -S(W_d^{\otimes k}y) +\left\|\bar{r}_{d}^{\otimes k} -W_d^{\otimes k}y \right\|_1\right]\\
    &= \inf_{y\in\R^{4^k}} \left[ \|y\|_1 +\left\|\bar{r}_{d}^{\otimes k} -W_d^{\otimes k}y \right\|_1\right]\,.
\ee
Here, in (i), we used that the infimum in \eqref{dual_prog} is achieved by taking 
\bb
    y_i=\max\left(0, \big(r_{d}^{\otimes k}-W_d^{\otimes k}\,(z-x)\big)_i\right)\,\,\,\,\forall\,i\in\{0,1,\ldots,4^k-1\}\,.
\ee
Moreover, (ii) easily follows by observing that the objective function evaluated at a given pair $(z,x)\in\R^{4^k}_+\times \R^{4^k}_+$ is always greater or equal to the objective function evaluated at the pair $((z-x)_+,(z-x)_-)\in\R^{4^k}_+\times \R^{4^k}_+$. In (iii), we introduced the vector $\bar{r}_d\coloneqq \frac{1}{2}r_d=\begin{pmatrix}
        \frac{1}{2d}, -\frac{1}{4}, \frac{d-1}{2d}, -\frac{1}{4}
        \end{pmatrix}^\intercal$.
In (iv), we denoted as $|x|$ the absolute value of a vector $x$, and we observed that $2x_+=x+|x|$. In (v), we employed that $S(|\cdot|)=\|\cdot\|_1$. In (vi), we used that $S(\bar{r}_d^{\otimes k})=\left(S(\bar{r}_d)\right)^k=0$ and that $S(W_d^{\otimes k}y)=S(y)$, where the latter easily follows by observing that $\sum_{i=0}^3(W_d)_{ij}=(1,1,1,1)^\intercal$ for all $j\in\{0,1,2,3\}$. This concludes the proof.
\end{proof}
The previous lemma reduces the PPT norm between the even and odd states to a minimisation problem of manageable form. 
To bound this quantity, we need to introduce some techniques in numerical analysis. First, let us recall a known result on the \emph{Tikhonov-regularised least squares problem}~\cite{tikhonov1977solutions}. Throughout, for a vector $x\in\R^n$ we denote its Euclidean norm by $\|x\|_2 \;\coloneqq\; \Bigl(\sum_{i=1}^n x_i^2\Bigr)^{1/2}$, and for a matrix $A\in\R^{n\times n}$ we write its singular value decomposition as $A=\sum_{i=1}^n \sigma_i\, u_i v_i^\intercal$, where $(\sigma_i)_{i=1}^n$ are the singular values, while $(u_i)_{i=1}^n$ and $(v_i)_{i=1}^n$ form orthonormal bases of $\R^n$.  In particular, $u_i$ (resp.~$v_i$) is an eigenvector of $AA^\intercal$ (resp.~$A^\intercal A$) with eigenvalue $\sigma_i^2$.
\begin{lemma}[(Tikhonov-regularised least squares~\cite{tikhonov1977solutions})]\label{lemma_tik}
Let $A\in\R^{n\times n}$ and $b\in\R^n$. Then
\bb
  \inf_{x\in\R^{n}}\Bigl(\|x\|_2^2+\|b-Ax\|_2^2\Bigr)
  \;=\; \sum_{i=1}^n \frac{(u_i^\intercal b)^2}{1+\sigma_i^2},
\ee
where $A=\sum_{i=1}^n \sigma_i\, u_i v_i^\intercal$ is the singular value decomposition of $A$. 
\end{lemma}

\begin{proof}
For completeness, we sketch the proof. Consider the objective function $f(x) \;=\; \|x\|_2^2 + \|b - Ax\|_2^2$. Differentiating with respect to $x$ and setting the gradient to zero shows that the minimiser is $\bar{x} \;=\; \sum_{i=1}^n \frac{\sigma_i}{\sigma_i^2+1}\,(u_i^\intercal b)\, v_i$. That is,  $\inf_{x\in\R^{n}}f(x) =\|\bar{x}\|_2^2 + \|b - A\bar{x}\|_2^2$. We now compute each term separately:
\bb
  \|\bar{x}\|_2^2 &= \sum_{i=1}^n \frac{\sigma_i^2}{(\sigma_i^2+1)^2}\,(u_i^\intercal b)^2\,,\\
  \|A\bar{x}-b\|_2^2
    &= \Biggl\|\sum_{i=1}^n \left(\frac{\sigma_i^2}{\sigma_i^2+1}-1\right)(u_i^\intercal b)\,u_i\Biggr\|_2^2
    = \sum_{i=1}^n \frac{1}{(\sigma_i^2+1)^2}\,(u_i^\intercal b)^2\,.
\ee
Adding the two contributions yields $\inf_{x\in\R^{n}}f(x) = \sum_{i=1}^n \frac{1}{1+\sigma_i^2}\,(u_i^\intercal b)^2$, which proves the claim.
\end{proof}

Second, we will use the celebrated \emph{Sanov's theorem}~\cite[Sec.~II.11]{CSISZAR-KOERNER}.  
Roughly speaking, Sanov's theorem quantifies how unlikely it is that the empirical distribution of i.i.d.~samples deviates significantly from the true distribution. More precisely, it shows that the probability of observing an empirical distribution inside a given set $\mathcal P$ decays exponentially fast in the number of samples, at a rate governed by the minimum relative entropy between an arbitrary distribution in $\mathcal P$ and the true distribution.
\begin{lemma}[{(Sanov's theorem~\cite[Exercise~2.12]{CSISZAR-KOERNER})}]\label{lem_sanov} 
Let $q=\{q_x\}_{x=1}^n$ be a probability distribution on an alphabet of $n$ elements $\{1,2,\ldots,n\}$, and let $X_1,\dots,X_k$ be $k$ i.i.d.~random variables drawn from $q$.  
The empirical distribution $\hat q^{(k)}$ is defined as
\bb
  \hat q^{(k)}_x \;\coloneqq\; \frac{1}{k}\,\#\{j: X_j=x\},
  \qquad x\in\{1,\dots,n\},
\ee
where $\#\{j: X_j=x\}$ denotes the number of occurrences of symbol $x$ among the $k$ samples.  
Let $\mathcal P$ be a set of probability distributions on $\{1,\dots,n\}$. Then
\bb
  \Pr\!\left[\hat q^{(k)}\in\mathcal P\right]
  \;\le\; (k+1)^n \, 2^{-k \min_{p\in\mathcal P} D(p\|q)},
\ee
where $D(p\|q) \;\coloneqq\; \sum_{x=1}^n p_x\bigl(\log_2 p_x - \log_2 q_x\bigr)$ is the relative entropy (Kullback-Leibler divergence) between $p$ and $q$.  
If the set $\mathcal P$ is convex, the prefactor can be removed, yielding the sharper bound
\bb\label{sharper_sanov}
  \Pr\!\left[\hat q^{(k)}\in\mathcal P\right]
  \;\le\; 2^{-k \min_{p\in\mathcal P} D(p\|q)}.
\ee
\end{lemma}
Specifically, we will use the following consequence of Sanov's theorem.
\begin{lemma}\label{conseq_sanov}
Let $q=(q_1,q_2,q_3)$ be a probability distribution with $q_2<q_3$, and consider the convex set
\bb
  \mathcal{P}\;\coloneqq\;\Bigl\{(p_1,p_2,p_3)\in\R_+^3:\; p_1+p_2+p_3=1,\;\; p_2\ge p_3\Bigr\}.
\ee
Then, for the empirical distribution $\hat{q}^{(k)}$ obtained from $k$ i.i.d.\ samples from $q$, we have
\bb
  \Pr\!\left[\hat{q}^{(k)}\in\mathcal{P}\right]
   \;\le\; \bigl(q_1+2\sqrt{q_2q_3}\bigr)^k.
\ee
\end{lemma}
\begin{proof}
Since $\mathcal{P}$ is convex, the sharpened version of Sanov's theorem (\eqref{sharper_sanov} of Lemma~\ref{lem_sanov}) applies:
\bb\label{eq_ssss}
  \Pr\!\left[\hat{q}^{(k)}\in\mathcal{P}\right]
   \;\le\; 2^{-k \min_{p\in\mathcal{P}} D(p\|q)}.
\ee
Thus it remains to compute 
\bb
  \min_{p\in\mathcal{P}} D(p\|q)
   \;=\; \min_{\substack{p_1,p_2,p_3\ge0 \\ p_1+p_2+p_3=1 \\ p_2\ge p_3}}
          D\bigl((p_1,p_2,p_3)\,\|\,(q_1,q_2,q_3)\bigr).
\ee
To identify the minimiser, consider perturbations of the form $(p_1,p_2-t,p_3+t)$ for $t\ge0$. Differentiating with respect to $t$ at $t=0$ gives
\bb
  \frac{\mathrm d}{\mathrm dt}\,
  D\!\left((p_1,p_2-t,p_3+t)\,\big\|\,(q_1,q_2,q_3)\right)\Big|_{t=0}
   \;=\; \log_2\!\left(\frac{p_3 q_2}{p_2 q_3}\right).
\ee
Since $q_2<q_3$, this derivative is strictly negative for all $(p_1,p_2,p_3)\in\mathcal{P}$. It follows that the minimum is attained for $p_2=p_3$, so the optimisation reduces to
\bb
  \min_{p\in[0,1]} D\!\left(\left(p,\tfrac{1-p}{2},\tfrac{1-p}{2}\right)\,\Big\|\,(q_1,q_2,q_3)\right).
\ee
A direct calculation shows
\bb
  D\!\left(\left(p,\tfrac{1-p}{2},\tfrac{1-p}{2}\right)\,\Big\|\,(q_1,q_2,q_3)\right)
   \;=\; D\!\left((p,1-p)\,\Big\|\,\Bigl(\tfrac{q_1}{q_1+2\sqrt{q_2q_3}},\;\tfrac{2\sqrt{q_2q_3}}{q_1+2\sqrt{q_2q_3}}\Bigr)\right)
          \;-\;\log_2\!\bigl(q_1+2\sqrt{q_2q_3}\bigr).
\ee
The minimum value is thus $\min_{p\in\mathcal{P}}D(p\|q)=-\log_2\left(q_1+2\sqrt{q_1q_2}\right)$. Substituting back into \eqref{eq_ssss}, we conclude the proof.
\end{proof}
We are now ready to provide an explicit upper bound on the minimisation problem from Lemma~\ref{lemma_lp}.  
\begin{lemma}\label{lemma_upper_bound}
For all $d,k\in\N$ with $d\ge 2$, it holds that
\bb\label{linear_program22}
   \inf_{x\in\R^{4^k}} \Bigl( \|x\|_1 \;+\; \left\|\bar{r}_d^{\otimes k}-W_d^{\otimes k}x\right\|_1 \Bigr) 
   \;\le\; 2\,\mu_d^k\,,
\ee
where $\bar{r}_d$ and $W_d$ are defined in~\eqref{rdwd}, and $\mu_d$ is defined in~\eqref{def_mud}.
\end{lemma}
\begin{proof} 
    It holds that
    \bb\label{chain_ineq}           \inf_{x\in\R^{4^k}}\left[\|x\|_1+\left\|\bar{r}_{d}^{\otimes k}-W_d^{\otimes k}x  \right\|_1\right]&\le
        2^k\inf_{x\in\R^{4^k}}\left[\|x\|_2+\left\|\bar{r}_{d}^{\otimes k}-W_d^{\otimes k}x  \right\|_2\right]\\
        &\le 2^{k+\frac12}\sqrt{\inf_{x\in\R^{4^k}}\left[\|x\|_2^2+\left\|\bar{r}_{d}^{\otimes k}-W_d^{\otimes k}x  \right\|_2^2\right]}\,,
    \ee
    where both inequalities follow from  $\|y\|_1\le \sqrt{d}\|y\|_2$ for all $d\in\mathbb{N}$ and $y\in\R^d$.

    We are now going to apply Lemma~\ref{lemma_tik} with $A\coloneqq W_d^{\otimes k}$ and $b\coloneqq \bar{r}_d^{\otimes k}$. To do so, we need to find a singular value decomposition for $W_d^{\otimes k}$. By denoting as $W_d=\sum_{i=1}^4\sigma_i u_iv_i^\intercal$  a singular value decomposition for $W_d$, it follows that a singular value decomposition for $W_d^{\otimes k}$ is given by $W_d^{\otimes k}=\sum_{\textbf{i}\in\{1,2,3,4\}^k} \sigma_\textbf{i} u_\textbf{i}v_\textbf{i}^\intercal$, where we defined
    \bb
        \sigma_\textbf{i}\coloneqq \sigma_{i_1}\sigma_{i_2}\ldots\sigma_{i_k}\,,\quad
        u_\textbf{i}\coloneqq u_{i_1}\otimes u_{i_2}\ldots \otimes u_{i_k}\,,\quad
        v_\textbf{i}\coloneqq v_{i_1}\otimes v_{i_2}\ldots \otimes v_{i_k}\,.
    \ee
    By performing the singular value decomposition of $W_d$ (see the Mathematica notebook attached), we obtain that:
    \bb\label{eq_sigma_i}
        \sigma_1=1\,,\quad
        \sigma_2=1\,,\quad
        \sigma_3=\sqrt{s}\,,\quad
        \sigma_4=\frac{1}{\sqrt{s}}\,,
    \ee
    where
    \bb
        s&\coloneqq \frac{16 - 8d + 4d^2 - 2d^3 + d^4 - ( d-2) \sqrt{64 + 32d^2 + 8d^4 + d^6}}{4d^2}\,,
    \ee
    and $u_1,u_2,u_3,u_4$ are orthonormal vectors defined as
    \bb\label{def_us}
        u_1 &\coloneqq \frac{1}{\sqrt{2}} 
        \begin{pmatrix}
            0 \\ 0 \\ 1 \\ 1
        \end{pmatrix}, \\
        u_2 &\coloneqq \frac{1}{\sqrt{3 + (1 + d)^2}} 
        \begin{pmatrix}
            1 \\ d+1 \\ -1 \\ 0
        \end{pmatrix}, \\
        u_3 &\propto 
        \begin{pmatrix}
            8 + 4d + 4d^2 + d^3 + d^4 + (d+1)\sqrt{64 + 32 d^2 + 8 d^4 + d^6}   \\[8pt]
            -8 - d (4 + d^2) - \sqrt{64 + 32 d^2 + 8 d^4 + d^6} \\[8pt]
            -4d \\[8pt]
            4d
        \end{pmatrix}, \\
        u_4 &\propto 
        \begin{pmatrix}
            8 + 4d + 4d^2 + d^3 + d^4 - (d+1)\sqrt{64 + 32 d^2 + 8 d^4 + d^6}  \\[8pt]
            -8 - d (4 + d^2) + \sqrt{64 + 32 d^2 + 8 d^4 + d^6} \\[8pt]
            -4d \\[8pt]
            4d
        \end{pmatrix}.
    \ee
    In particular, it holds that  $W_dW_d^\intercal=u_1u_1^\intercal+u_2u_2^\intercal+s u_3u_3^\intercal+\frac{1}{s} u_4u_4^\intercal$. Moreover, by taking the inverse of the above equation and exploiting the fact that $W_d^{-1}=W_d$ (which simply follows from \eqref{Theta_partial_transpose} because the partial transposition is an involution), it follows that
    \begin{equation}\label{eq_WtW}
        W_d^\intercal W_d=u_1u_1^\intercal+u_2u_2^\intercal+s^{-1} u_3u_3^\intercal+s u_4u_4^\intercal\,.
    \end{equation}\
    Hence, the singular value decomposition of $W_d$ is of the form $W_d= u_1u_1^\intercal+u_2u_2^\intercal+\sqrt{s}u_3u_4^\intercal+\frac{1}{\sqrt{s}}u_4u_3^\intercal$,     and thus it holds that $v_1=u_1$, $v_2=u_2$, $v_3=u_4$, and $v_4=u_3$.  Consequently, we deduce that
\begin{align}
        &\frac12 \left\|\rho_0^{(k,d)}-\rho_1^{(k,d)}\right\|_{\ppt} \nonumber\\
        &\quad \leqt{(ii)}\   2^{k+\frac12}\sqrt{\inf_{x\in\R^{4^k}}\left[\|x\|_2^2+\left\|\bar{r}_{d}^{\otimes k}-W_d^{\otimes k}x  \right\|_2^2\right]} \nonumber\\ 
        &\quad \eqt{(iii)}\  2^{k+\frac12}\sqrt{ \sum_{\textbf{i}\in\{1,2,3,4\}^k}\frac{1}{\sigma_{\textbf{i}}^2+1}\left[(\bar{r}^{\otimes k})^\intercal u_{\textbf{i}}\right]^2  } \nonumber\\
        &\quad \eqt{(iv)}\ 2^{k+\frac12}\sqrt{ \sum_{\textbf{i}\in\{1,2,3,4\}^k}\frac{1}{s^{\#_3(\textbf{i})-\#_4(\textbf{i})}+1}\left[(\bar{r}^{\otimes k})^\intercal u_{\textbf{i}}\right]^2  } \nonumber\\
        &\quad \eqt{(v)}\ 2^{k+\frac12}\sqrt{ \sum_{\textbf{i}\in\{2,3,4\}^k}\frac{1}{s^{\#_3(\textbf{i})-\#_4(\textbf{i})}+1} c_{\textbf{i}}  } \nonumber\\
        &\quad =\ 2^{k+\frac12}\sqrt{ \sum_{\textbf{i}\in\{2,3,4\}^k}\frac{1}{s^{\#_3(\textbf{i})-\#_4(\textbf{i})}+1} c_2^{\#_2(\textbf{i})}c_3^{\#_3(\textbf{i})}c_4^{\#_4(\textbf{i})}  } \nonumber\\
        &\quad \le\ 2^{k+\frac12}\sqrt{ \sum_{\textbf{i}\in\{2,3,4\}^k}s^{\max(0,\#_4(\textbf{i})-\#_3(\textbf{i}))} c_2^{\#_2(\textbf{i})}c_3^{\#_3(\textbf{i})}c_4^{\#_4(\textbf{i})}  } \nonumber\\
        &\quad =\ 2^{k+\frac12}\sqrt{ \sum_{\substack{\textbf{i}\in\{2,3,4\}^k\\ \#_4(\textbf{i})\ge \#_3(\textbf{i}) }}c_2^{\#_2(\textbf{i})}\left(\frac{c_3}{s}\right)^{\#_3(\textbf{i})}(sc_4)^{\#_4(\textbf{i})}+\sum_{\substack{\textbf{i}\in\{2,3,4\}^k\\ \#_4(\textbf{i})<\#_3(\textbf{i})}} c_2^{\#_2(\textbf{i})}c_3^{\#_3(\textbf{i})}c_4^{\#_4(\textbf{i})}   } \label{eq_proof_main_lemma} \\
        &\quad \eqt{(vi)}\ 2^{k+\frac12}\sqrt{ \sum_{\substack{\textbf{i}\in\{2,3,4\}^k\\ \#_4(\textbf{i})\ge \#_3(\textbf{i}) }}c_2^{\#_2(\textbf{i})}c_4^{\#_3(\textbf{i})}c_3^{\#_4(\textbf{i})}+\sum_{\substack{\textbf{i}\in\{2,3,4\}^k\\ \#_4(\textbf{i})<\#_3(\textbf{i})}} c_2^{\#_2(\textbf{i})}c_3^{\#_3(\textbf{i})}c_4^{\#_4(\textbf{i})}   } \nonumber\\
        &\quad \le\ 2^{k+\frac12}\sqrt{ \sum_{\substack{\textbf{i}\in\{2,3,4\}^k\\ \#_4(\textbf{i})\ge \#_3(\textbf{i}) }}c_2^{\#_2(\textbf{i})}c_4^{\#_3(\textbf{i})}c_3^{\#_4(\textbf{i})}+\sum_{\substack{\textbf{i}\in\{2,3,4\}^k\\ \#_4(\textbf{i})\le\#_3(\textbf{i})}} c_2^{\#_2(\textbf{i})}c_3^{\#_3(\textbf{i})}c_4^{\#_4(\textbf{i})}   } \nonumber\\
        &\quad =\ 2^{k+1}\sqrt{ \sum_{\substack{\textbf{i}\in\{2,3,4\}^k\\ \#_4(\textbf{i})\le\#_3(\textbf{i})}} c_2^{\#_2(\textbf{i})}c_3^{\#_3(\textbf{i})}c_4^{\#_4(\textbf{i})}   } \nonumber\\
        &\quad =\ 2^{k+1}\sqrt{ \left( c_2+c_3+c_4\right)^k\sum_{\substack{\textbf{i}\in\{2,3,4\}^k\\ \#_4(\textbf{i})\le\#_3(\textbf{i})}} \left(\frac{c_2}{c_2+c_3+c_4}\right)^{\#_2(\textbf{i})}\left(\frac{c_3}{c_2+c_3+c_4}\right)^{\#_3(\textbf{i})}\left(\frac{c_4}{c_2+c_3+c_4}\right)^{\#_4(\textbf{i})}  } \nonumber\\
        &\quad \leqt{(vii)}\ 2\left[2\sqrt{c_2+2\sqrt{c_3c_4}}\right]^{k} \nonumber\\
        &\quad \eqt{(viii)\,}\ 2 \left(\sqrt{\,1- \frac{ \frac{5}{8}+\frac{1}{d}\left(\tfrac{1}{4}+\tfrac{2}{d}+\tfrac{9}{d^{2}}-\tfrac{6}{d^{3}} -\;\sqrt{2}\,\Bigl(\tfrac{9}{4}+\tfrac{3}{d}+\tfrac{1}{d^{2}}\Bigr) \sqrt{1-\frac{2}{d+\frac{4}{d}}}\right) }{ 1+\tfrac{2}{d}+\tfrac{4}{d^{2}} }\,}\right)^k \\
 &\quad =\ 2\mu_d^k\nonumber\,.
 \end{align}
    Here, in (ii), we exploited the inequality in~\eqref{chain_ineq}. In (iii), we used     Lemma~\ref{lemma_tik}. In (iv), we leveraged \ref{eq_sigma_i} to observe that $\sigma_{\textbf{i}}^2=s^{\#_3(\textbf{i})-\#_4(\textbf{i})}$, where we denoted as  $\#_j(\textbf{i})$ the total number of $j$'s among the elements of the string $\textbf{i}$. In (v), we defined for all $\textbf{i}\in\{2,3,4\}^k$ the quantity $c_{\textbf{i}}$ as $c_{\textbf{i}}\coloneqq c_{i_1}c_{i_2}\ldots c_{i_n}$, where
    \bb\label{def_cs}
        c_2\coloneqq (\bar{r}_d^\intercal u_1)^2+(\bar{r}_d^\intercal u_2)^2\,,\qquad
        c_3\coloneqq (\bar{r}_d^\intercal u_3)^2\,,\qquad
        c_4\coloneqq (\bar{r}_d^\intercal u_4)^2\,.
    \ee
    In (vi), we observed that $s=\frac{c_3}{c_4}$. The latter can be proved either by a direct calculation or as follows. Note that
    \bb\label{eq_for_s}
        c_2+c_3+c_4=\bar{r}_d^\intercal\bar{r}_d= \bar{r}_d^\intercal W_d^\intercal W_d\bar{r}_d=c_2+\frac{1}{s}c_3+sc_4\,,
    \ee
    where the first equality comes from \eqref{def_cs} and from the fact that $(u_1,u_2,u_3,u_4)$ are orthonormal, the second equality is a consequence of the fact that $W_d\bar{r}_d=\bar{r}_d$ (which can be proved either by a direct calculation or by exploiting \eqref{Theta_partial_transpose} together with the fact that $\sigma_0^{(d)}-\sigma_1^{(d)}$ is invariant under partial transposition, as proved in the Appendix~\ref{sec_appendix_c}), and the third equality follows by \ref{eq_WtW}. Hence, by rearranging \ref{eq_for_s}, we obtain that $s=\frac{c_3}{c_4}$.
    In (vii), we applied the consequence of the Sanov theorem stated in Lemma~\ref{conseq_sanov}. Specifically, we observed that 
    \bb
        P&\coloneqq \sum_{\substack{\textbf{i}\in\{2,3,4\}^k\\ \#_4(\textbf{i})\le\#_3(\textbf{i})}} \left(\frac{c_2}{c_2+c_3+c_4}\right)^{\#_2(\textbf{i})}\left(\frac{c_3}{c_2+c_3+c_4}\right)^{\#_3(\textbf{i})}\left(\frac{c_4}{c_2+c_3+c_4}\right)^{\#_4(\textbf{i})}\,
    \ee
    is exactly the probability that the empirical distribution $\hat{q}^{(k)}$, after $k$ samples extracted by the probability distribution $q=(q_2,q_3,q_4)$ defined as 
\bb
    q_2\coloneqq  \frac{c_2}{c_2+c_3+c_4}\,,\qquad
    q_3=\frac{c_3}{c_2+c_3+c_4}\,\qquad
    q_4=\frac{c_4}{c_2+c_3+c_4}\,,
\ee
is contained in the set of probability distributions $\mathcal{P}$ defined as
\bb
    \mathcal{P}\coloneqq \{(p_2,p_3,p_4)\in\R_+^3:\quad p_2+p_3+p_4=1,\quad p_3\ge p_4\}\,.
\ee
Hence, by employing Lemma~\ref{conseq_sanov}, it follows that 
\bb
    P\le \left(q_2+2\sqrt{q_3q_4}\right)^k = \frac{\left(c_2+\sqrt{c_3c_4}\right)^k}{\left(c_2+c_3+c_4\right)^k}\,,
\ee
which proves (vii) in \eqref{eq_proof_main_lemma}. Finally, in (viii), we explicitly calculated the term $2\sqrt{c_2+2\sqrt{c_3c_4}}$ by exploiting \eqref{def_cs} and \eqref{def_us} (see the Mathematica notebook attached). This concludes the proof.
\end{proof}

\subsection{Concluding the proof}\label{sec_proof_of_prop_main}
We are now ready to assemble the preceding lemmas and establish Proposition~\ref{main_PROP}.  
Recall that the proposition asserts that, for all $d,k\in\N$ with $d\ge 2$, the LOCC norm between the even and odd states satisfies
\bb
  \tfrac12\,\bigl\|\rho_1^{(k,d)}-\rho_0^{(k,d)}\bigr\|_{\locc}
  \;\le\; 2\,\mu_d^k,
\ee
where $\mu_d$ is defined in~\eqref{def_mud}.

\begin{proof}[Proof of Proposition~\ref{main_PROP}]
It holds that
\bb
    \frac12\,\bigl\|\rho_1^{(k,d)}-\rho_0^{(k,d)}\bigr\|_{\locc}&\leqt{(i)} \frac12\,\bigl\|\rho_1^{(k,d)}-\rho_0^{(k,d)}\bigr\|_{\ppt}\\
    &\eqt{(ii)} \inf_{x\in\R^{4^k}} \Bigl( \|x\|_1 \;+\; \left\|\bar{r}_d^{\otimes k}-W_d^{\otimes k}x\right\|_1 \Bigr)\\
    &\leqt{(iii)} 2\mu_d^k\,,
\ee
where: in (i) we employed the general fact that the LOCC norm is upper bounded by the PPT norm; in (ii)  we applied Lemma~\ref{lemma_lp}; and in (iii) we used Lemma~\ref{lemma_upper_bound}. This concludes the proof. 
\end{proof}

\section{Conclusions}\label{sec_conclusion}
In this work we resolved an open problem in the theory of quantum data hiding, specifically establishing the existence of bipartite states that are simultaneously separable, perfectly distinguishable under global operations, and yet nearly indistinguishable under LOCC measurements. In other words, we provided an explicit scheme to achieve quantum data hiding with orthogonal and separable states. Our construction proceeds in two steps: first, we identify two separable, orthogonal states that are not perfectly distinguishable under LOCC; second, we amplify their indistinguishability by considering multiple copies and applying a parity-based encoding. Concretely, we proved the existence of separable, orthogonal $\varepsilon$-quantum data hiding states on $\mathbb{C}^D\otimes \mathbb{C}^D$, where the local dimension scales as $D = O(1/\varepsilon^{10})$, while any such construction must necessarily satisfy $D = \Omega(1/\varepsilon)$.

A compelling direction for future research is to sharpen the dependence of the local dimension on $\varepsilon$, closing the gap between the current $O(1/\varepsilon^{10})$ upper bound and the $\Omega(1/\varepsilon)$ lower bound. Another natural open question is to prove or disprove Conjecture~\ref{cj1}, which would imply that the parity construction applied to any pair of states that are not perfectly distinguishable via LOCC automatically yields quantum data hiding states. Proving this conjecture would immediately provide a broad class of new examples of separable, orthogonal quantum data hiding states.

\medskip\medskip

\begin{note}
The central result of this work, the existence of perfectly orthogonal data hiding states, was announced in a seminar at the Free University of Berlin in the Summer 2023. Our explicit estimates on the local dimension required to achieve data hiding were derived at the beginning of February 2025. While writing up this paper, we became aware of~\cite{Ha2025}, whose main result is similar to ours. The proof techniques in the two papers, however, are significantly different.
\end{note}

\section{Acknowledgements}
We are deeply indebted to Bartosz Regula for suggesting the name `biblical state' for $\sigma_1^{(d)}$. FAM and LL thank the Free University of Berlin for its hospitality in 2023, during which this project was initiated. FAM and LL acknowledge financial support from the European Union (ERC StG ETQO, Grant Agreement no.\ 101165230). Views and opinions expressed are however those of the authors only and do not necessarily reflect those of the European Union or the European Research Council. Neither the European Union nor the granting authority can be held responsible for them. FAM~acknowledges financial support from the project: PRIN 2022 ``Recovering Information in Sloppy QUantum modEls (RISQUE)'', code 2022T25TR3, CUP E53D23002400006.

\bibliography{biblio}

\appendix

\section{$\GG$-twirling}\label{sec_appendix_a}
Construct the group of $d\times d$ unitaries
\bb
\GG \coloneqq \left\{ U_\pi V_\e:\ \pi \in S_d\, ,\ \e\in \{\pm 1\}^d \right\} ,
\ee
where $U_\pi \coloneqq \sum_{i=0}^{d-1}\ketbraa{\pi(i)}{i}$ implements a permutation $\pi$ in the symmetric group over $d$ elements $S_d$, and $V_\e \coloneqq \sum_{i=0}^{d-1} \e_i \ketbra{i}$ is a diagonal Hermitian unitary. Consider the $\GG$-twirling
\bb
\TT_\GG (X) \coloneqq \frac{1}{|\GG|} \sum_{U\in \GG} (U\otimes U)\, X\, (U\otimes U)^\dag\, .
\label{G_twirling2}
\ee
\begin{lemma}
An alternative expression for the $\GG$-twirling~\eqref{G_twirling} is
\bb \label{G_twirling_simplified2}
\TT_\GG(X) = \sum_{i=0}^3 \frac{\Tr[X\,\Theta_i]}{\Tr\Theta_i}\, \Theta_i\,,
\ee where $\Theta_0,\Theta_1,\Theta_2,\Theta_3$ are the four mutually orthogonal projectors in~\eqref{def_theta}.
\end{lemma}
\begin{proof}
It can be easily verified that the four operators $\Theta_0,\Theta_1,\Theta_2,\Theta_3$ commute with unitaries of the form $U\otimes U$, where $U\in \GG$. It can also be checked that these are the only four linearly independent operators that have this property. Without embarking on a complicated ad hoc reasoning, there is a standard way of doing so, which is that of counting the irreps of the representation $\GG\ni U\mapsto U \otimes U$. We can do so with the theory of characters:
\begin{align}
\frac{1}{|\GG|} \sum_{U\in \GG} (\Tr U)^4 &\eqt{(i)} 
\frac{1}{d!\, 2^d} \sum_{k=0}^d \left( \frac{d!}{k!} \sum_{\ell=0}^{d-k} \frac{(-1)^\ell}{\ell!} \right) \sum_{\e\in \{\pm 1\}^d} \left(\sumno_{j=1}^k \e_j\right)^4 \nonumber \\
&\eqt{(ii)} \sum_{k=0}^d \left( \frac{1}{k!} \sum_{\ell=0}^{d-k} \frac{(-1)^\ell}{\ell!} \right) \left(3k^2-2k\right) \nonumber \\
&\eqt{(iii)} \sum_{m=0}^d \sum_{k=0}^m \frac{(-1)^{m-k}}{(m-k)!\, k!} \left(3k^2-2k\right) \nonumber \\
&= \sum_{m=0}^d \frac{1}{m!} \sum_{k=0}^m \binom{m}{k} (-1)^{m-k} \left(3k^2-2k\right) \\
&= \sum_{m=0}^d \frac{1}{m!} \sum_{k=0}^m \binom{m}{k} (-1)^{m-k} \left(3 (t\partial_t)^2 - 2 t\partial_t\right) t^k \big|_{t=1} \nonumber \\
&= \sum_{m=0}^d \frac{1}{m!} \left(3 (t\partial_t)^2 - 2 t\partial_t\right) (t-1)^m \Big|_{t=1} \nonumber \\
&= \sum_{m=0}^d \frac{1}{m!} \left(3m \delta_{m,2} + \delta_{m,1} \right) \nonumber \\
&= 4\, . \nonumber
\end{align}
Here, in~(i) we remembered that there are exactly $\frac{d!}{k!} \sum_{\ell=0}^{d-k} \frac{(-1)^\ell}{\ell!}$ permutations of $d$ elements that fix exactly $k$ arbitrary elements, in~(ii) we noticed that there are precisely $k^2 + k(k-1)2 = 3k^2-2k$ ways of picking four elements in $\{1,\ldots,k\}$ such that one can form two pairs of equal elements,\footnote{If the first two elements are equal, and there are $k$ ways this can happen, then the second must also be made of equal elements, yielding a total of $k^2$ choices. If the first two elements are different, and this can happen in $k(k-1)$ ways, then there are only two choices for the second pair.} and finally in~(iii) we introduced the new parameter $m\coloneqq k+\ell$, which ranges between $0$ and $d$.

The above calculation tells us that the four operators we have found above are the only ones that commute with all unitaries of the form $U\otimes U$, where $U\in \GG$. Hence, the $\GG$-twirling in~\eqref{G_twirling} must act as in~\eqref{G_twirling_simplified}.
\end{proof}

\section{LOCC distinguishability of the two special states}\label{sec_appendix_b}
In this section we analyse the distinguishability of the states $\sigma_0^{(d)}$ and $\sigma_1^{(d)}$ introduced in Definition~\ref{def_sigma01} under restricted classes of measurements. In particular, Proposition~\ref{locc_norm_xi2} provides an exact evaluation of both the PPT norm and the \emph{separable norm}~\cite{VV-dh} between these states, as well as upper and lower bounds on their LOCC norm. 
\begin{prop}\label{locc_norm_xi2}
We have that
\bb
\frac12 - \frac{1}{d} \leq \frac12\left\| \sigma_0^{(d)} - \sigma_1^{(d)} \right\|_\locc \leq \frac12\left\| \sigma_0^{(d)} - \sigma_1^{(d)} \right\|_\sep = \frac12 \left\| \sigma_0^{(d)} - \sigma_1^{(d)} \right\|_\ppt = \frac12 + \frac1d\, .
\ee
\end{prop}
\begin{proof}
Setting $k=1$ in Lemma~\ref{lemma_lp} and considering the ansatz
\bb
x_0 \coloneqq \left( \frac{3d-2}{4d(d-1)},\, 0,\, \frac{d-2}{4d},\, 0 \right)^\intercal ,
\ee
we obtain that  
\bb
\frac12 \left\| \sigma_0^{(d)} - \sigma_1^{(d)} \right\|_\sep \leq \frac12 \left\| \sigma_0^{(d)} - \sigma_1^{(d)} \right\|_\ppt = \inf_{x\in \R^4} \left( \|x\|_1 + \left\| \widebar{r}_d - W_d x\right\|_1 \right) \leq \|x_0\|_1 + \left\| \widebar{r}_d - W_d x_0\right\|_1  = \frac12 + \frac1d\,, 
\ee
where the first inequality follows from the general fact that a separable measurement is also PPT~\cite{VV-dh}. For the lower bound on the separable norm, we can consider the POVM operator $E\coloneqq P - \Phi + \frac2d Q_-$, where $P$, $\Phi$, and $Q_-$ are defined in \eqref{def_orth_proj}. It turns out that $E^\Gamma \geq 0$ and $(\id - E)^\Gamma \geq 0$, so that $(E, \id-E)$ is a PPT measurement --- as a matter of fact, it is also separable~\cite[Section~4]{Park2024}. Hence,
\bb
\frac12 \left\| \sigma_0^{(d)} - \sigma_1^{(d)} \right\|_\ppt \geq \frac12 \left\| \sigma_0^{(d)} - \sigma_1^{(d)} \right\|_\sep \geq \Tr\left[ E \left(\sigma_1^{(d)} - \sigma_0^{(d)}\right)\right] = \frac12 + \frac1d\, .
\ee
Since the separable norm always upper bounds the LOCC norm~\cite{VV-dh}, the only claim that remains to be shown is the lower bound on the LOCC norm. The simple LOCC protocol of measuring both subsystems in the computational basis and checking whether the two outcomes coincide yields
\bb
\frac12\left\| \sigma_0^{(d)} - \sigma_1^{(d)} \right\|_\locc &\geq \Tr\left[ \sum_{i=0}^{d-1}\ketbra{i}\otimes\ketbra{i}\, \left(\sigma_1^{(d)} - \sigma_0^{(d)}\right) \right] = \frac12 - \frac1d\, ,
\ee
concluding the proof.
\end{proof}
\section{Invariance of the two special states under partial transposition}\label{sec_appendix_c}

\begin{lemma}
The states $\sigma_0^{(d)}$ and $\sigma_1^{(d)}$ defined in Definition~\ref{def_sigma01} are invariant under partial transposition. That is, $(\sigma_0^{(d)})^\Gamma = \sigma_0^{(d)}$ and $(\sigma_1^{(d)})^\Gamma = \sigma_1^{(d)}$.
\end{lemma}
\begin{proof}
Recall that 
\bb
    \sigma_0^{(d)} \;=\; \tfrac{1}{d}\,\Theta_0 + \tfrac{2}{d^2}\,\Theta_2,
    \qquad
    \sigma_1^{(d)} \;=\; \tfrac{1}{2(d-1)}\,\Theta_1 + \tfrac{1}{d(d-1)}\,\Theta_3,
\ee
where the projectors $\Theta_i$ satisfy
\bb\label{eq_theta_gamma}
    \Theta_i^\Gamma \;=\; \sum_{j=0}^3 (W_d)_{ij}\,\Theta_j,
    \qquad \forall\, i\in\{0,1,2,3\},
\ee
with $W_d$ the matrix defined in~\eqref{rdwd}.  Substituting the decomposition~\eqref{eq_theta_gamma} into the expressions of $\sigma_0^{(d)}$ and $\sigma_1^{(d)}$, and using the explicit form of $W_d$, one verifies directly that
$(\sigma_0^{(d)})^\Gamma = \sigma_0^{(d)}$ and $(\sigma_1^{(d)})^\Gamma = \sigma_1^{(d)}$. This proves the claim.
\end{proof}

\end{document}